\documentclass[12pt]{article}
\usepackage{geometry} % see geometry.pdf on how to lay out the page. There's lots.
\usepackage{amsfonts}
\usepackage{amsmath}
\usepackage{fullpage}

\newtheorem{theorem}{Theorem}

\newtheorem{proposition}[theorem]{Proposition}

\newenvironment{proof}[1][Proof]{\begin{trivlist}
\item[\hskip \labelsep {\bfseries #1}]}{\end{trivlist}}

\begin{document}

\vspace{4cm}
\noindent
{\bf \Large Jack polynomial fractional quantum Hall states and their generalizations}

\vspace{5mm}
\noindent
Wendy Baratta and
Peter J.~Forrester

\noindent
Department of Mathematics and Statistics,
University of Melbourne, \\
Victoria 3010, Australia  \\

\small
\begin{quote}
In the the study of fractional quantum Hall states, a certain clustering condition involving
up to four integers has been identified. We give a simple proof that particular Jack polynomials with
$\alpha = - (r-1)/(k+1)$, $(r-1)$ and $(k+1)$ relatively prime, and with partition given in terms of
its frequencies by $[n_00^{(r-1)s}k 0 ^{r-1}k 0 ^{r-1}k \cdots  0 ^{r-1}  m]$ satisfy this clustering condition. 
Our proof makes essential use of the fact that these Jack polynomials are translationally invariant. We also consider nonsymmetric Jack polynomials, symmetric and nonsymmetric generalized Hermite and Laguerre polynomials, and Macdonald polynomials from the viewpoint of the clustering.
\end{quote}

\section{Introduction}

The symmetric Jack polynomials $P_{\kappa}(z;\alpha)$, $z:=(z_1,\ldots,z_N)$ a
coordinate in $\mathbb C^N$, $\alpha$ a scalar and $\kappa$ a partition, are an orthogonal, 
homogeneous basis for symmetric function generalizing the Schur ($\alpha=1$) and zonal ($\alpha=2$) polynomials. They appear in physics in random 
matrix theory \cite[Ch. 12 \& 13]{Fo10}, \cite{De08} and in the study of quantum many body wave functions
 \cite[Ch. 11]{Fo10}, \cite{BH08a,BH08}.
Here we will be interested in the latter interpretation.

There are two classes of quantum many body systems for which Jack polynomials are relevant,
one involving the $1/r^2$ pair potential in one dimension and the other corresponding to certain fractional quantum Hall states. Regarding the former  \cite[Ch. 11]{Fo10},
with the domain a unit circle, the corresponding Sch\"{o}dinger operator reads
\begin{equation}\label{HC}
H^{(C)}:=-\sum_{j=1}^N\frac{\partial^2}{\partial \theta^2_j}+\frac{\beta}{4}\bigg(\frac{\beta}{2}-1\bigg)\sum_{1\leq j <k \leq N}\frac{1}{\sin^2(\theta_k-\theta_j)/2},
\end{equation}
where  $\beta$ parametrizes the coupling.
With $z_j=e^{i\theta_j}$ the ground state wave function for (\ref{HC}) is proportional to 
\begin{equation}\label{J1}
\psi_0^{(C)}(z):=|\Delta(z)|^{\beta/2},\qquad \Delta(z):=\prod_{1\leq j <k \leq N}(z_j-z_k)
\end{equation}
and with $\alpha:=2/\beta$, a complete set of eigenfunctions is given in terms of Jack polynomials by \cite[eq. (13.199)]{Fo10}
\begin{equation}\label{J1a}
\psi_0^{(C)}(z)z^{-l}P_{\kappa}(z;\alpha)\qquad (l=0,1,\ldots)
\end{equation}
where
\begin{equation}\label{J1b}
z^\kappa:=z_1^{\kappa_1}z_2^{\kappa_2}\ldots z_N^{\kappa_N}
\end{equation}
and for $l>0$ it is required that $\kappa_N=0$.

Next we will revise how certain fractional quantum Hall states relate to Jack polynomials \cite{BH08a,BH08}.
 An infinite family of bosonic  fractional quantum Hall states states, indexed by a positive integer $k$, are due to Read and Rezayi \cite{RR99}. For a system of $kN$ particles, these are defined up to normalization as
\begin{equation}\label{RR}
\psi_{RR}^{(k)}={\rm Sym} \prod_{s=1}^k\prod_{1\leq i_s <j_s \leq N}(z_{i_s}-z_{j_s})^2
\end{equation}
where Sym denotes symmetrization (see (\ref{21}) below). Note that the $kN$ particles are thus partitioned into $k$ groups of $N$. Setting $k=1$ we read off that
$$
\psi_{RR}^{(1)}=\prod_{1\leq j<k\leq N}(z_j-z_k)^2
$$
which is the filling factor $\nu=1/2$ bosonic Laughlin state. For $k=2$ it turns out that \cite{RR99}
\begin{equation}\label{Pf}
\psi_{RR}^{(2)}={\rm Pf}\Big[\frac{1}{z_k-z_l}\Big]_{k,l=1,\ldots,2N} \prod_{1\leq i<j \leq 2N}(z_{i}-z_{j}),
\end{equation}
where the diagonal entry is to be replaced by zero if $k=l$, which is the filling factor $\nu=1$ Moore-Read state \cite{MR91}. 
As noted in \cite{RR99}, $\psi_{RR}^{(k)}$ is
characterized by the requirements that it be symmetric, and exhibit the factorization property
\begin{equation}\label{RRz}
\psi_{RR}^{(k)}(z_1,\ldots,z_{(N-1)k},\underbrace{z,\ldots,z}_{k \;\; {\rm  times}}\,)=\prod_{l=1}^{(N-1)k}(z_l-z)^2\psi_{RR}^{(k)}(z_1,\ldots,z_{(N-1)k}).
\end{equation}
It is at this stage the Jack polynomials show themselves. Thus it is a remarkable finding of recent times \cite{BH08a} that (\ref{RRz}) is satisfied by
\begin{equation}\label{J3a}
\psi_{RR}^{(k)}(z)=P_{(2\delta)^k}(z;-k-1),
\end{equation}
where $\delta:=(N-1,N-2,\ldots,1,0)$, $2\delta$ means each part of $\delta$ is multiplied by $2$, and $(2\delta)^k$ means each part of $2\delta$ is repeated $k$ times.

The relation (\ref{J3a}) is one result in a broader theory relating Jack polynomials to quantum Hall states. This comes about by generalizing (\ref{RRz}) to the so called $(k,r)$ clustering property 
\cite{BH08}
\begin{equation}\label{J4}
\psi^{(k, r)}(z_1,\ldots,z_{(N-1)k},\underbrace{z,\ldots,z}_{k \;\; {\rm  times}}\,)=\prod_{l=1}^{N-k}(z_l-z)^r\psi^{(k,r)}(z_1,\ldots,z_{(N-1)k})
\end{equation} 
(see also
\cite{WW08,LWWW09} in relation to general factorizations of quantum Hall states).
For $k=1$ and $r$ even this, together with the requirement $\psi^{(k,r)}$ be symmetric, implies
\begin{equation}\label{J4a}
\psi^{(1,r)}(z_1,\ldots,z_N)=\prod_{1\leq j<k \leq N}(z_j-z_k)^r
\end{equation}
which is the filling factor $\nu=1/r$ bosonic Laughlin state. For $k=2$,  $N\mapsto 2N$ and
$r$ odd it is known \cite{BH08a} that (\ref{J4}) is satisfied by 
$$
\psi^{(2,r)}(z_1,\ldots,z_{2N})={\rm Pf}\Big[\frac{1}{z_k-z_l}\Big]_{k,l=1,\ldots,2N}\prod_{1\leq i<j \leq 2N}(z_i-z_j)^r
$$
which is the $\nu=1/r$ Moore-Read state.

For general $k\in \mathbb{Z}^+$ it was conjectured in \cite{BH08a} and later proved in \cite{ES09} using methods from conformal field theory 
(see also the related works \cite{BGS09,ERS10,EBS10})
that for $k+1$ and $r-1$
relatively prime, (\ref{J4}) is satisfied by
\begin{equation}\label{J51}
\psi^{(k,r)}(z_1,\ldots, z_N)=P_{\kappa(k,r)}(z_1,\ldots, z_N;-(k+1)/(r-1)),
\end{equation}
where $\kappa(k,r)$ is the staircase partition \cite{JL10}
\begin{equation}\label{J52}
(((\beta+1)r+1)^k,(\beta r+1)^k,\ldots,(r+1)^k).
\end{equation}
In (\ref{J52}) $\beta\in \mathbb{Z}^+$ must be related to $N$ by 
\begin{equation}\label{J53}
N=\frac{k+1}{r-1}+k(\beta+2),
\end{equation}
and the notation $(\kappa_1^{n_1},\kappa_2^{n_2},\ldots,\kappa_p^{n_p})$ means that the part $\kappa_1$ is repeated $n_1$ times, $\kappa_2$ is repeated $n_2$ times etc. Alternatively, if $f_j$ denotes the frequency of the part equal to $j$ in $\kappa$ (e.g. if $\kappa=211100$ then $f_0=2,\,f_1=3,\,f_2=1$), $\kappa$ is specified in terms of its frequencies according to \cite{BH08a}
\begin{equation}\label{J54}
\kappa(k,r)=[k 0^{r-1}k 0^{r-1}k 0^{r-1}k\ldots].
\end{equation}

We see from (\ref{J52}) or (\ref{J54}) that 
\begin{equation}\label{4.1}
\kappa_i-\kappa_{i+k}\geq r
\end{equation}
which in \cite{BH08a} was interpreted as a generalized exclusion principle. The significance of such partitions in Jack polynomial theory was first noticed by Feigin et al. \cite{FJMM02}, who showed that the set
of Jack polynomials $\{P_{\kappa(k,r)+\mu}(z;-(k+1)/(r-1)) \}_\mu$ forms a basis for the set of symmetric functions vanishing when $k+1$ variables coincide.

The Laughlin, Read-Moore and Read-Rezayi states are all translationally invariant, and so satisfy
\begin{equation}\label{Lp}
L^+\psi=0,\qquad L^+:=\sum_{j=1}^N\frac{\partial}{\partial z_j}.
\end{equation}
This can also be interpreted as a highest weight condition in a raising and lowering operator formalism of angular momentum on the sphere, projected onto the plane \cite{LL97}. The companion lowest weight condition is that 
\begin{equation}\label{Lm}
\bigg(\sum_{j=1}^Nz_j^2 \frac{\partial^2}{\partial z_j^2}+N_\phi \sum_{j=1}^Nz_j\bigg)\psi=0
\end{equation}
where $N_\phi$ is interpreted as the monopole charge. When (\ref{Lp}) and (\ref{Lm}) are satisfied, as is the case for Laughlin, Read-Moore and Read-Rezayi states, $N_\phi$ must obey
$$
\sum_{j=1}^Nz_j\frac{\partial}{\partial z_j} \psi=\frac{N}{2}N_\phi\psi.
$$
Most importantly, the Jack polynomials (\ref{J51}) satisfy both (\ref{Lp}) and (\ref{Lm}) and so are well founded quantum Hall states.

The contribution to the study of these so called Jack states in this paper relates to viewing the clustering condition (\ref{J4}), and related factorization formulas, as identities in Jack polynomial theory. We know that symmetric Jack polynomial theory has a number of extensions and generalizations. In particular there are multivariable classical orthogonal polynomials which appear in the study of the eigenfunctions of variants of the Calogero-Sutherland Schr\"{o}dinger operator \cite{BF97b}; there are nonsymmetric versions of the Jack polynomials and the multivariable classical orthogonal polynomials \cite{Ch95, BF98b}; and there are $q$-generalizations by way of Macdonald polynomial theory \cite[Ch.~VI]{Ma95}. It is our aim to initiate a study of the clustering condition (\ref{J4}), and related factorizations in the context of these additional families of polynomials. 

In Section 2 we revise Jack polynomial theory, and its extensions to generalized Hermite and Laguerre polynomials, and Macdonald polynomials, as needed for use in subsequent sections. In Section 3
we show that 
the case $k=1$ of the clustering (\ref{J4}) can be solved in terms of Jack polynomials involving an
arbitrary partition (actually this is an already known result). We provide too a similar solution in
terms of nonsymmetric Jack polynomials, symmetric and nonsymmetric generalized Hermite and Laguerre polynomials, and symmetric Macdonald polynomials (the latter after an appropriate
$(q,t)$-generalization).

In Section 4 we provide a very simple proof that (\ref{J51}) satisfies (\ref{J4}). The main ingredient in our proof is the fact that the Jack polynomials (\ref{J4}) are translationally invariant. This proof also applies to the more general clustering (\ref{25.1}) below, first isolated in \cite{BH08}. We show that the symmetric
generalized Hermite and Laguerre polynomials coincide with the Jack polynomials under the conditions that the latter satisfy (\ref{25.1}). We provide a Macdonald polynomial analogue of (\ref{25.1}), but we do not have a proof.

\section{Preliminary theory}
\subsection{Jack polynomials}
Let $\kappa:=(\kappa_1,\ldots,\kappa_N)$ denote a partition of non-negative integers such that $\kappa_1\geq\kappa_2\geq\ldots\geq \kappa_N$ and $|\kappa|:=\sum_{j=1}^N\kappa_j$ be its modulus, $l(\kappa)$ its length (i.e. number of non-zero parts) and define $z^\kappa$ as in (\ref{J1b}). The monomial symmetric functions $m_\kappa(z)$ are specified by
$$
m_\kappa(z)=\frac{1}{C}{\rm Sym} \, z^{\kappa}
$$
where 
\begin{equation}\label{21}
{\rm Sym}f(z_1,\ldots, z_N)=\sum_{\sigma \in S_N}f(z_{\sigma(1)},\ldots, z_{\sigma(N)})
\end{equation}
and the normalization $C$ is chosen so that the coefficient of $z^\kappa$ in $m_\kappa$ is unity.

Continuing with the definitions, let $<$ be a partial ordering on partitions $|\mu|=|\kappa|$, $\mu\not=\kappa$, specified by $\mu<\kappa$ iff
$$
\sum_{j=1}^p \mu_i \leq \sum_{j=1}^p\kappa_j \qquad (p=1,\ldots,N).
$$
The symmetric Jack polynomial $P_\kappa(z;\alpha)$, labelled by a partition $\kappa$ and dependent on a scalar parameter $\alpha$, can be specified as the polynomial eigenfunction of the differential operator
\begin{equation}\label{H1}
\widetilde{H}^{(C)}:=\sum_{j=1}^N\Big(z_j\frac{\partial}{\partial z_j} \Big)^2+\frac{2}{\alpha}\sum_{1\leq j<k\leq N}\frac{z_j+z_k}{z_j-z_k} \Big(\frac{\partial}{\partial z_j}-\frac{\partial}{\partial z_k} \Big)
\end{equation}
with eigenvalue
\begin{equation}\label{H1a}
e(\kappa;\alpha)=\sum_{j=1}^N\kappa_j(\kappa_j-1)+(\alpha(N-1)+1)|\kappa|-2\alpha\sum_{j=1}^N(j-1)\kappa_j,
\end{equation}
and having the structure
\begin{equation}
P_{\kappa}(z;\alpha)=m_\kappa(z)+\sum_{\mu<\kappa}a_{\kappa \mu}m_\mu(z)
\end{equation}
for some coefficients $a_{\kappa \mu}\in \mathbb{Q}(\alpha)$. In fact $P_{\kappa}(z;\alpha)$ is the unique symmetric polynomial eigenfunction of $\widetilde{H}_N^{(C)}$ with leading term $m_\kappa(z)$
and eigenvalue (\ref{H1a}). We remark too that $\widetilde{H}^{(C)}$ is related to the Calogero-Sutherland operator (\ref{HC}) by
\begin{equation}\label{HH}
|\Delta(z)|^{-1/\alpha}(H^{(C)}-E_0^{(C)})|\Delta(z)|^{1/\alpha}=\widetilde{H}^{(C)},
\end{equation}
where $E_0^{(C)}$ is the ground state energy.

Fundamental to the theory of the integrability properties of (\ref{HC}) is the more general Schr\"{o}dinger operator
\begin{equation}\label{HCE}
H^{(C,Ex)}=-\sum_{j=1}^N\frac{\partial^2}{\partial \theta_j^2}+\frac{\beta}{4}\sum_{1\leq j < k \leq N}\frac{(\beta/2-s_{jk})}{\sin^2(\theta_j-\theta_k)/2}
\end{equation}
where $s_{jk}$ is the operator acting on a function $f(\theta_1,\ldots,\theta_N)$ by interchanging $\theta_j$ and $\theta_k$. The ground state wave function is again proportional to (\ref{J1}). Conjugating by this state as in (\ref{HH}) gives the transformed operator
\begin{align}
\widetilde{H}^{(C,Ex)}=&|\Delta(z)|^{-1/\alpha}(H^{(C,Ex)}-E_0^{(C)})|\Delta(z)|^{1/\alpha} \notag \\
=&\sum_{j=1}^N\Big(z_j\frac{\partial}{\partial z_j}\Big)^2+\frac{(N-1)}{\alpha}\sum_{j=1}^Nz_j\frac{\partial}{\partial z_j } \notag\\
&+\frac{2}{\alpha}\sum_{1\leq j <k \leq N}\frac{z_jz_k}{z_j-z_k}\bigg(\Big(\frac{\partial}{\partial z_j }-\frac{\partial}{\partial z_k } \Big) - \frac{1-s_{jk}}{z_j-z_k} \bigg) \label{HCE1}.
\end{align} 
The significance of (\ref{HCE}) shows itself upon the introduction of the mutually commuting Cherednik operators \cite{Ch95}, \cite[Def. 11.4.3]{Fo10}
\begin{equation}\label{xi}
\xi_i:=\alpha z_i d_i+1-N+\sum_{p=i+1}^Ns_{ip} \qquad (i=1,\ldots,N),
\end{equation}
where $d_i$ denotes the type $A$ Dunkl operator \cite{Du89}, \cite[Def. 11.4.2]{Fo10}
\begin{equation}\label{di}
d_i:=\frac{\partial}{\partial z_i}+\frac{1}{\alpha}\sum_{\substack{k=1 \\ \not= i}}^N\frac{1-s_{jk}}{z_i-z_k}.
\end{equation}
Thus
$$
\widetilde{H}^{(C,Ex)}=\frac{1}{\alpha^2}\sum_{j=1}^n\Big( \xi_i+\frac{N-1}{2} \Big)^2 -E_0^{(C)}.
$$

With $\eta$ denoting a composition $\eta=(\eta_1,\ldots,\eta_N)$ $(\eta_j\in\mathbb{Z}_{\geq0})$, $\{\xi_i\}$ permits a complete set of simultaneous polynomial eigenfunctions $\{ E_{\eta}(z;\alpha)\}_{\eta},$
$$
\xi_i E_\eta(z;\alpha)=\overline{\eta}_i E_\eta(z;\alpha) \qquad (1=1,\ldots,N),
$$
where the eigenvalue $\overline{\eta}_i$ is specified by
\begin{equation}\label{ni}
\overline{\eta}_i=\alpha \eta_i - \#\{k<i|\eta_k\geq \eta_i\}-\#\{k>i|\eta_k>\eta_i\}.
\end{equation}
The $E_\eta(z;\alpha)$ are referred to as the nonsymmetric Jack polynomials, and analogous to (\ref{H1a}) they exhibit the structure
\begin{equation}\label{di1}
E_\eta(z;\alpha)=z^\eta+\sum_{\nu \prec \eta}\widetilde{a}_{\eta \nu} z^\nu.
\end{equation}
With $\rho^+$ denoting the partition corresponding to the composition $\rho$,
in (\ref{di1}) $\prec$ denotes the Bruhat ordering on compositions, defined by the statement that $\nu\prec\eta$ if $\nu^+\prec \eta^+$, or in the case $\nu^+=\eta^+$, if $\nu=\prod_{l=1}^rs_{i_lj_l}\eta$ where $\eta_{i_l}>\eta_{j_l}$, $i_l<j_l$.

The operator (\ref{HCE}) also permits a complete set of symmetric, and anti-symmetric, polynomial eigenfunctions. The complete set of symmetric eigenfunctions are the symmetric Jack polynomials. Since Sym commutes with (\ref{HCE1}) we have 
\begin{equation}\label{J5a}
{\rm Sym} \, E_\eta(z;\alpha)=a_\eta P_{\eta^+}(z;\alpha)
\end{equation}
for some $a_\eta$ (see \cite[eq. (12.101)]{Fo10}.

The complete set of anti-symmetric polynomial eigenfunctions of (\ref{HCE1}) are referred to as the anti-symmetric Jack polynomials \cite{BF99}. They are denoted $S_{\kappa+\delta}(z;\alpha)$ where $\delta$ is as in (\ref{J3a}), and have the structure 
$$
S_{\kappa+\delta}(z;\alpha)=\Delta(z)\Big(m_\kappa+\sum_{\sigma<\kappa}\widehat{a}_{\kappa \sigma}m_\sigma\Big)
$$
(cf. (\ref{H1a})). With
$$
{\rm Asym} \, f(z_1,\ldots, z_N):=\sum_{P\in S_N}\epsilon(P)f(z_{P(1)},\ldots,z_{P(N)}),
$$
where $\epsilon(P)$ denotes the signature of $P$,
analogous to (\ref{J5a}), for $\rho^+=\kappa+\delta$, we have
\begin{equation}\label{J5b}
{\rm Asym} \, E_\rho(z)=c_\rho S_{\kappa+\delta}(z;\alpha)
\end{equation}
for some $c_\rho$ (see \cite[eq. (12.113)]{Fo10}). Furthermore, the symmetric and anti-symmetric Jack polynomials are related by \cite[eq. (12.118)]{Fo10}
\begin{equation}\label{J6a}
S_{\kappa+\delta}(z;\alpha)=\Delta(z)P_{\kappa}(z;\alpha/(1+\alpha)).
\end{equation}

\subsection{Generalized classical polynomials}
The quantum many body system on a circle with $1/r^2$ pair potential, as specified by the Schr\"{o}dinger operator (\ref{HC}), can also be defined on a line with an harmonic confining potential. When generalized to include exchange terms the Schr\"{o}dinger operator for the latter reads \cite[Prop. 11.3.1]{Fo10}
\begin{equation}\label{517}
H^{(H,Ex)}:=-\sum_{j=1}^N\frac{\partial^2}{\partial x_j^2}+\frac{\beta^2}{4}\sum_{j=1}^Nx_j^2+\beta\sum_{1\leq j< k \leq N}\frac{\beta/2-s_{jk}}{(x_j-x_k)^2}.
\end{equation}
This has ground state wave function proportional to
\begin{equation}\label{518}
\psi_0^{(H)}(x)=\prod_{l=1}^Ne^{-\beta x_l ^2/4}\prod_{1\leq j<k\leq N}|x_k-x_j|^{\beta/2},
\end{equation}
and furthermore permits a complete set of eigenfunctions of the form 
$$
\psi_0^{(H)}(x)E_\eta^{(H)}(\sqrt{\beta/2}x;\alpha),
$$
where $\{E_\eta^{(H)}(y;\alpha) \}$ are referred to as the generalized nonsymmetric Hermite polynomials. These polynomials are eigenfunctions of the transformed operator
\begin{align}
\widetilde{H}^{(H,Ex)}&:=-\frac{2}{\beta}(\psi_0^{(H)}(x))^{-1}(H^{(H,Ex)}-E_0^{(H)})\psi_0^{(H)}(x) \notag \\
&=\sum_{j=1}^N\Big(\frac{\partial^2}{\partial y_j^2} -2 y_j \frac{\partial}{\partial y_j}\Big)+\frac{2}{\alpha}\sum_{j<k}\frac{1}{y_j-y_k}\bigg( \Big(\frac{\partial}{\partial y_j}-\frac{\partial}{\partial y_k}\Big)-\frac{1-s_{jk}}{y_j-y_k}\bigg), \label{HEx}
\end{align}
where $E_0^{(H)}$ denotes the ground state energy and we have changed variables $y_j=\sqrt{\beta/2}x_j$.

The operator (\ref{HEx}) permits a decomposition in terms of the generalized Laplacian
$$
\Delta_A:=\sum_{i=1}^Nd_i^2,
$$
where $d_i$ denotes the Dunkl operator (\ref{di}). With the $d_i$ defined in terms of $\{y_i\}$, a direct calculation shows \cite[Prop. 11.5.1]{Fo10}
\begin{equation}\label{J7a}
\Delta_A=\widetilde{H}^{(H,Ex)}+2\sum_{j=1}^Ny_j\frac{\partial}{\partial y_j}.
\end{equation}
Moreover, $\Delta_A$ can be used to generate the $E_\eta^{(H)}$ from the nonsymmetric Jack polynomials according to \cite[eq. (13.91)]{Fo10}
\begin{equation}\label{J7}
{\rm exp} \Big(-\frac{1}{4}\Delta_A \Big)E_\eta(y;\alpha)=E_\eta^{(H)}(y;\alpha).
\end{equation}
And with symmetric $P_\kappa^{(H)}$ and anti-symmetric $S_{\kappa+\delta}^{(H)}$ generalized Hermite polynomials constructed from the $E_\eta^{(H)}$ by the analogue of (\ref{J5a}) and (\ref{J5b}), the appropriate modification of (\ref{J7}) generates these polynomials from their symmetric and anti-symmetric counterparts.

It is well known \cite{OP83} that (\ref{517}) is related to the $A$ type root system. There is also a Calogero-Sutherland system on the half line $x\geq0$ with $B$ type symmetry (unchanged by $x\mapsto -x$), specified by the Schr\"{o}dinger operator
\begin{align}
 H^{(L,Ex)}:=&-\sum_{j=1}^N\frac{\partial^2}{\partial x_j^2}+\frac{\beta^2}{4}\sum_{j=1}^Nx_j^2+\frac{(\beta a + 1)}{2}\sum_{j=1}^N\frac{(\beta a +1) /2-\sigma_j}{x_j^2} \notag\\
 &+\beta\sum_{1\leq j<k\leq N}\bigg(\frac{\beta/2-s_{jk}}{(x_j-x_k)^2}+\frac{\beta/2-\sigma_j \sigma_k s_{jk}}{(x_j+x_k)^2} \bigg). \label{11.54}
 \end{align}
Here  $\sigma_j$ is the operator which replaces the coordinate $x_j$ by $-x_j$.

The ground state wave function is proportional to 
$$
\psi_0^{(L)}(x^2)=\prod_{l=1}^Nx_l^{(\beta a + 1)/2}e^{-\beta x_l^2/4}\prod_{1\leq j < k \leq N}|x_k^2-x_j^2|^{\beta/2}
$$
and there is a complete set of even eigenfunctions of the form \cite{BF98b}
$$
\psi_0^{(L)}(x^2)E_{\eta}^{(L)}\Big( \frac{\beta}{2}x^2;\alpha\Big),
$$
where $\{ E_\eta^{(L)}(y^2;\alpha)\}$ are referred to as the generalized nonsymmetric Laguerre polynomials. The latter are eigenfunctions of the transformed operator
\begin{align}
\widetilde{H}^{(L,Ex)}:=&\frac{2}{\beta}(\psi_0^{(L)})^{-1}(H^{(L,Ex)}-E_0^{(L)})\psi_0^{L}(x^2) \notag \\
=&\frac{1}{4}\sum_{j=1}^N\Big(\frac{\partial^2}{\partial y_j^2}-2y_j\frac{\partial}{\partial y_j}+(2a+1)\frac{1}{y_j}\frac{\partial}{\partial y_j} \Big) \notag 
\\
&+\frac{1}{\alpha}\sum_{j<k}\frac{1}{y_j^2-y_k^2}
\bigg( \Big( y_j\frac{\partial}{\partial y_j}-y_k\frac{\partial}{\partial y_k}\Big)
-\frac{y_j^2+y_k^2}{y_j^2-y_k^2}(1-s_{jk})\bigg),
 \label{LEx}
\end{align}
where as in (\ref{HEx}) we have changed variables $y_j=\sqrt{\beta/2}x_j$. Analogous to (\ref{J7a}), with 
$$
d_i^{(B)}:=\frac{\partial}{\partial y_i}+\frac{1}{\alpha}\sum_{p=1}^N\Big( \frac{1-s_{ip}}{y_i-y_p}+\frac{1-\sigma_i\sigma_p s_{ip}}{y_i+y_p}\Big)+\frac{(a+1/2)}{y_i}(1-\sigma_i)
$$
we have \cite[eq. (11.78)]{Fo10}
\begin{equation}\label{LEy}
\Delta_B:=\sum_{i=1}^N(d_i^{(B)})^2=4 \Big(\widetilde{H}^{(L,Ex)}+\frac{1}{2}\sum_{j=1}^Ny_j\frac{\partial}{\partial y_j}\Big),
\end{equation}
provided $\Delta_B$ is restricted to act on functions even in each $y_j$. We can use this operator to compute the nonsymmetric Laguerre polynomials in terms of the nonsymmetric Jack polynomials according to \cite[eq. (13.116)]{Fo10}
\begin{equation}\label{J9}
{\rm exp}\Big( -\frac{1}{4} \Delta_B \Big)E_\eta(y^2;\alpha)=E_\eta^{(L)}(y^2;\alpha).
\end{equation}

\subsection{Macdonald polynomials}
Macdonald polynomials \cite{Ma95} generalize Jack polynomials. They were introduced into the study of fractional quantum Hall states in the recent work \cite{JL10}.

The symmetric Macdonald polynomials $P_{\kappa}(z;q,t)$
can be uniquely characterized as the symmetric polynomial solutions of the eigenvalue equation
\begin{equation}\label{Ma}
M_1 P_\kappa(z;q,t)=e(\kappa;q,t)P_\kappa(z;q,t),
\end{equation}
with a structure the same as exhibited in (\ref{H1a}) for the symmetric Jack polynomials. Here
\begin{equation}\label{Mb}
M_1:=\sum_{i=1}^N\prod_{\substack{j=1\\\not=i}}^N\frac{tz_i-z_j}{z_i-z_j}T_{q,z_i},
\end{equation}
where $T_{q,z_i}$ acts on $f(z_1,\ldots, z_N)$ by the replacement $z_i \mapsto q z_i$, and the eigenvalue has the explicit form
\begin{equation}\label{Mc}
e(\kappa;q,t)=\sum_{i=1}^Mq^{\kappa_i}t^{N-i}.
\end{equation}
They relate to the Jack polynomials by
\begin{equation}\label{Md}
\lim_{q\rightarrow 1}P_\kappa(z;q,q^{1/\alpha})=P_\kappa(z;\alpha).
\end{equation}

The nonsymmetric Macdonald polynomials $E_\eta(z;q,t)$ can be characterized as the simultaneous polynomial eigenfunctions
$$
Y_iE_\eta(z;q,t)=\overline{\eta}_iE_\eta(z;q,t)\qquad (1=1,\ldots,N)
$$
with structure as in (\ref{di1}). Here
$$
Y_i:=t^{-N+i}T_i\ldots T_{N-1}\omega T_1^{-1}\ldots T_{i-1}^{-1}
$$
where, with $s_i:=s_{i,i+1}$
\begin{align*}
T_i&:=t+\frac{tz_i-z_{i+1}}{z_i-z_{i+1}}(s_i-1), \\
\omega&:=s_{n-1}\ldots s_2 T_{q,z_1},
\end{align*} 
and 
$$
\overline{\eta}_i := q^{\eta_i} t^{-l_\eta'(i)}, \quad
l_\eta'(i)=\#\{j<i|\eta_j\geq \eta_i\}-\#\{j>i|\eta_j>\eta_i\}. 
$$

Introducing the $t$-symmetrization and $t$-antisymmetrization operators by
\begin{equation}\label{UU}
U^+:=\sum_{\sigma\in S_N}T_\sigma,\qquad U^-:=\sum_{\sigma\in S_N}\Big(-\frac{1}{t}\Big)^{l(\sigma)}T_\sigma,
\end{equation}
where $\sigma:=s_{i_l(\sigma)}\ldots s_{i_1}$ is a minimal length decomposition in terms of transpositions, and $T_\sigma:=T_{i_l(\sigma)}\ldots T_{i_1}$, we have \cite{Ma99}
\begin{equation}\label{Me}
U^+E_\eta(z;q,t)=a_\eta(q,t)P_{\eta^+}(z;q,t)
\end{equation}
for some (known) $a_\eta(q,t)$. Also, with $S_{\kappa+\delta}(z;q,t)$ defined as the $t$-antisymmetric polynomial eigenfunctions of the eigenvalue equation 
$$
\Big(\sum_{i=1}^N Y_i\Big) S_{\kappa+\delta}(z;q,t)=\Big( \sum_{i=1}^N\overline{\eta}_i\Big) S_{\kappa+\delta}(z;q,t)
$$
we have
$$
U^- E_{\delta+\eta}(x;q,t)=b_\eta(q,t)S_{\delta+\eta^+}(z;q,t)
$$
for some (known) $b_\eta(q,t)$. And analogous to (\ref{J6a}) these $t$-antisymmetric Macdonald polynomials are related to their symmetric counterparts by \cite{Ma99}
\begin{equation}\label{5Pa}
S_{\delta+\kappa}(z;q,t)=t^{-N(N-1)/2}\Delta_t(z)P_\kappa(z;q,qt),
\end{equation}
where
$$
\Delta_t(z):=\prod_{1\leq j<k \leq N}(t z_j-z_k).
$$
With $t=q^{1/\alpha}$ we see from (\ref{Md}) that in the limit $q\rightarrow 1$ (\ref{5Pa}) reduces to (\ref{J6a}).

\section{The clustering condition for $k=1$}
\subsection{Symmetric Jack polynomials}
Iterating (\ref{J4}) with $k=1$ gives (\ref{J4a}). This is a symmetric function for $r$ even, and an anti-symmetric function for $r$ odd. According to (\ref{J51}), for $r$ even we have 
\begin{equation}\label{J8.1}
\prod_{1\leq j < k \leq N} (z_j-z_k)^r=P_{r\delta}(z_1,\ldots,z_N;-2/(r-1)),
\end{equation}
where the partition $r\delta$ is specified as in (\ref{J3a}). In the context of fractional quantum Hall states, this result was first proved in \cite{BH08}. The method used was to observe that for the product formula (\ref{J4a})
$$
D_i \psi^{(1,r)}(z)=0, \qquad D_i:=\frac{\partial}{\partial z_i}-r\sum_{\substack{j=1 \\ \not=i}}^N\frac{1}{z_i-z_j}.
$$
Consequently
$$
\sum_{j=1}^N(z_jD_j)^2\psi^{(1,r)}(z)=0.
$$
On the other hand, by direct calculation 
$$
\sum_{j=1}^N(z_jD_j)^2=\Big(\widetilde{H}^{(C)}-e(r\delta;\alpha) \Big) \bigg|_{\alpha=-2/(r-1)},
$$
and we know that the unique symmetric polynomial null vector of this latter operator with leading term $z^{r\delta}$ is the Jack polynomial in (\ref{J8.1}).

In fact the result (\ref{J8.1}) is a special case of a more general identity, already known in the Jack polynomial literature \cite{Op98}, \cite[Ex. 12.6 q.5]{Fo10}.

\begin{proposition}\label{Prop1}
Let $r>0$ be even. For a general partition $\kappa$, $l(\kappa)\leq N$,
\begin{equation}\label{578}
P_{r\delta+\kappa}(z;-2/(r-1))=(\Delta(z))^rP_\kappa(z;2/(r+1)).
\end{equation}
\end{proposition}
\begin{proof}
According to the theory below (\ref{HC})
\begin{equation}\label{579}
(H^{(C)}-E_0^{(C)})\Big( |\Delta(z)|^{1/\alpha}P_\kappa(z;\alpha) \Big)=|\Delta(z)|^{1/\alpha}P_\kappa(z;\alpha).
\end{equation}
Suppose we replace $1/\alpha$ by $1-1/\alpha$ in this equation. Making use of the simple but crucial observation that
$$
H^{(C)}\bigg|_{1/\alpha \mapsto 1-1/\alpha}=H^{(C)}
$$
then applying (\ref{HH}) we see
\begin{equation}\label{580}
\widetilde{H}^{(C)}|\Delta(z)|^{1-2/\alpha}P_\kappa(z;1/(1-1/\alpha))=e(\kappa;1/(1-1/\alpha))|\Delta(z)|^{1-2/\alpha}P_\kappa(z;1/(1-1/\alpha)).
\end{equation}
The next step is to replace $1/\alpha$ by $-\alpha+1/2$, replace $\kappa$ by $\kappa+(\alpha(N-1))^N$ and to use the basic property of Jack polynomials
$$
P_{\kappa+p^N}(z;\alpha)=z^{p^N}P_\kappa(z;\alpha)
$$
to deduce from (\ref{580}) that for $\alpha$ a non-negative integer
\begin{align}
&\widetilde{H}^{(C)}\bigg|_{\alpha \mapsto 1/(-\alpha+1/2)}\Big((\Delta(z))^{2\alpha}P_\kappa(z;1/(\alpha+1/2))\Big) \notag \\
&\hspace{4cm}=e(\kappa+(\alpha(N-1))^N;1/(\alpha+1/2))(\Delta(z))^{2\alpha}P_\kappa(z;1/(\alpha+1/2)) \label{J10}.
\end{align}
Furthermore, we can check from the definition (\ref{H1}) that
\begin{equation}\label{J10a}
e(\kappa+(\alpha(N-1))^N;1/(\alpha+1/2))=e(\kappa+2\alpha\delta;1/(-\alpha+1/2)).
\end{equation}
This tells us that (\ref{J10}) is the eigenequation for the Jack polynomial $P_{\kappa+2\alpha\delta}(z;1/(-\alpha+1/2))$. We know too that the latter is the unique polynomial eigenfunction of this
eigenequation with leading term $z^{\kappa+2\alpha\delta}$ so (\ref{578}) follows.
\hfill $\square$\end{proof}

\subsection{Nonsymmetric Jack polynomials}
A feature of (\ref{578}) is that $r$ must be even. It turns out that in the case of the nonsymmetric Jack polynomials the analogue of $r$ must be odd.

\begin{proposition}\label{Prop2}
For $l>0$ and odd, and $\kappa$ a partition with $l(\kappa)\leq N$ we have 
\begin{equation}\label{J11}
E_{\kappa+l\delta}(z;-2/l)=(\Delta(z))^{l}E_\kappa(z;2/l).
\end{equation}
\end{proposition}
\begin{proof}
Both sides have the same leading term $z^{\delta l+ \kappa}$. Hence it suffices to show that for $l$ odd $(\Delta(z))^{l}E_\kappa(z;2/l)$ is a simultaneous polynomial eigenfunction of each $\xi_i |_{\alpha=-2/l}$ $(i=1,\ldots, N)$ with eigenvalue $(\overline{\kappa_i+l \delta_i}) |_{\alpha=-2/l}$. Now, we see from (\ref{di}) that for $l$ odd
\begin{align*}
&-\frac{2}{l}z_i d_i\bigg|_{\alpha=-2/l}\bigg((\Delta(z))^l E_\kappa(z;2/l) \bigg) \\
&\hspace{2cm}=-\frac{2}{l}(\Delta(z))^l z_i d_i\bigg|_{\alpha=2/l} E_\kappa(z;2/l) \\
&\hspace{2cm}=(\Delta(z))^l \Big(-\frac{2}{l}z_i\frac{\partial}{\partial z_i}-\sum_{\substack{k=1 \\ \not= i}}^N\frac{(1-s_{ik})}{z_i-z_k} \Big) E_\kappa(z;2/l) 
\end{align*}
Substituting in (\ref{xi}), again making essential use of $l$ being odd, shows
\begin{align*}
&\xi_i\bigg|_{\alpha=-2/l}\Big( (\Delta(z))^l E_\kappa(z;2/l) \Big) \\
&\hspace{2cm}=-(\Delta(z))^l\Big(\xi_i\bigg|_{\alpha=2/l}-2(1-N) \Big) E_\kappa(z;2/l) \\
&\hspace{2cm}=-\bigg(\overline{\kappa}_i\bigg|_{\alpha=2/l}-2(1-N)\bigg)(\Delta(z))^lE_\kappa(z;2/l)\\
&\hspace{2cm}=(\overline{\kappa_i+l\delta_i})\bigg|_{\alpha=-2/l}(\Delta(z))^lE_\kappa(z;2/l),
\end{align*}
where the final equality follows from (\ref{ni}), which is the sought equation.
\hfill $\square$\end{proof}

Note that after writing $l=r-1$, ($r$ even), $\kappa\mapsto\kappa+\delta$ and symmetrizing both sides (\ref{J11}) reads
\begin{equation}\label{12.1}
(\Delta(z))^{r-1}{\rm Asym} \;E_{\kappa+\delta}(z;2/(r-1))={\rm Sym}\; E_{r \delta +\kappa}(z;2/(r-1)).
\end{equation}
Recalling (\ref{J5b}), (\ref{J5a}) and (\ref{J6a}) we see that this reclaims (\ref{578}).

\subsection{Generalized Hermite and Laguerre polynomials}

We will first show that the nonsymmetric generalized Hermite and Laguerre polynomials satisfy a factorization formula structurally identical to (\ref{J11}). By symmetrization we can then deduce the analogues of (\ref{578}).

\begin{proposition}\label{K}
Let $l>0$ be odd, and let $\kappa$ be a partition $l(\kappa)\leq N$. We have
\begin{align}
E_{\kappa+l\delta}^{(H)}(z;-2/l)&=(\Delta(z))^{l}E^{(H)}_\kappa(z;2/l) \label{13.1}, \\
E^{(L)}_{\kappa+l\delta}(z;-2/l)&=(\Delta(z))^{l}E^{(L)}_\kappa(z;2/l).\label{13.2}
\end{align}
Also, with $r$ even and positive 
\begin{align}
P_{\kappa+l\delta}^{(H)}(z;-2/(r-1))&=(\Delta(z))^{r}P^{(H)}_\kappa(z;2/(r+1)), \label{13.3} \\
P^{(L)}_{\kappa+l\delta}(z;-2/(r-1))&=(\Delta(z))^{r}P^{(L)}_\kappa(z;2/(r+1)).\label{13.4}
\end{align}
\end{proposition}
\begin{proof}
First, we can note (\ref{13.3}), (\ref{13.4}) follow from (\ref{13.1}), (\ref{13.2}) by symmetrizing both sides to deduce the analogue of (\ref{12.1}), then using the Hermite and Laguerre analogues of (\ref{J5a}), (\ref{J5b}) and (\ref{J6a}).

To derive (\ref{13.1}) and (\ref{13.2}) we use the explicit forms of $\Delta_A$ and $\Delta_B$ as given in (\ref{J7}), (\ref{LEy}), (\ref{LEx}), together with the identity 
$$
\frac{1}{(a-b)(a-c)}+\frac{1}{(b-a)(b-c)}+\frac{1}{(c-a)(c-b)}=0
$$
to show that for $l$ odd
$$
\Delta_A\Big |_{\alpha=-2/l}\Big((\Delta(z))^lf(z) \Big)=(\Delta(z))^l \Delta_A\Big |_{\alpha=2/l}f(z)
$$
and 
$$
\Delta_B\Big |_{\alpha=-2/l}\Big((\Delta(z^2))^lf(z^2) \Big)=(\Delta(z^2))^l \Delta_B\Big |_{\alpha=2/l}f(z^2).
$$
This shows that when applying the exponential operators in (\ref{J7}) to (\ref{J11}), the nonsymmetric Jack polynomials therein are transferred to their generalized Hermite and Laguerre counterparts, as stated in (\ref{13.1}) and (\ref{13.2}).
\hfill $\square$\end{proof}

Setting $\kappa=0^N$ in Proposition \ref{K} shows that for $l$ odd
\begin{equation}\label{14.1}
E_{l\delta}(z;-2/l)=E_{l \delta}^{(H)}(z;-2/l)=E_{l\delta}^{(L)}(z;-2/l)=(\Delta(z))^l
\end{equation}
and for $r$ even
\begin{equation}\label{14.2}
P_{r\delta}(z;-2/(r-1))=P_{r \delta}^{(H)}(z;-2/(r-1))=P_{r\delta}^{(L)}(z;-2/(r-1))=(\Delta(z))^r.
\end{equation}
The first of these formulas dramatically displays the special properties which take effect for $\alpha=-2/l$, $l$ odd. Thus the generalized Hermite and Laguerre polynomials generally have the structure in terms of the Jack polynomials
\begin{align*}
E_\eta^{(H)}(z;\alpha)&=E_\eta(z;\alpha)+\sum_{|\nu|<|\eta|}b_{\eta \nu}E_\nu(z;\alpha) \\
E_\eta^{(L)}(z;\alpha)&=E_\eta(z;\alpha)+\sum_{|\nu|<|\eta|}\widetilde{b}_{\eta \nu}E_\nu(z;\alpha).
\end{align*}
However for $\alpha=-2/l$, $l$ odd and a special choice of $\eta$, all but the leading term vanishes and the generalized Hermite and Laguerre polynomials become homogeneous.
Similar remarks apply in relation to (\ref{14.2}).

\subsection{Symmetric Macdonald polynomials}
Let us set
\begin{equation}\label{DD}
D_l(z;q)=\prod_{i=1}^lD_1(z;q^{2i+1}),\qquad D_1(z;q):=\prod_{i=1}^{N}\prod_{\substack{j=1\\j\not=i}}^N(qz_j-z_i).
\end{equation}
For all $i,j\in\{1,\ldots,N \},$ $i\not=j$ and $0\leq s \leq 2l$ we see that $D_l(z;q)=0$ when
\begin{equation}\label{20.1}
z_j=z_it q^s,\qquad t=q^{-(2l-1)}.
\end{equation}
This is (a special case of) the wheel condition, first identified in \cite{FJMM03} in relation to the vanishing properties of Macdonald polynomials with $t^{k+1}q^{r-1}=1$, and partitions with parts satisfying
(\ref{4.1}). For $k=1$ and $r$ even the theory of \cite{FJMM03} tells us that the corresponding Macdonald polynomial with partition satisfying (\ref{4.1}) vanishes when (\ref{20.1}) holds with $q\mapsto q^{1/2}$ and $l=r/2$ ($r$ even) and thus that they contain $D_{r/2}(z;q^{1/2})$ as a factor. Consistent with the Jack polynomial identity (\ref{578}), the remaining factor is another Macdonald polynomial.

\begin{proposition}\label{PD}
Let $D_l$ be as in (\ref{DD}), and let $r>0$ be even. We have
\begin{equation}\label{20.c}
P_{\kappa+r\delta}(z;q,q^{-(r-1)/2})=(-q^{-1/2})^{r^2N(N-1)/8}D_{r/2}(z;q^{1/2})P_\kappa(z;q,q^{(r+1)/2}).
\end{equation}
\end{proposition}
\begin{proof}
Our strategy is to make use of the characterization of the Macdonald polynomials as eigenfunctions of (\ref{Ma}). We first observe that both sides of (\ref{20.c}) are polynomials with leading term the monomial symmetric function $m_{\kappa+r\delta}(z)$. It remains then to check that the RHS satisfies the same eigenvalue equation as the LHS.

Noting that 
$$
T_{q,z_i}D_{r/2}(z;q^{1/2})=q^{r/2}\frac{(q^{(r+1)/2}z_i-z_j)}{(q^{(r-1)/2}z_i-z_j)}D_{r/2}(z;q^{1/2}),
$$
and recalling the definition (\ref{Mb}), we see
\begin{align}
&M_1\bigg|_{t=q^{-(r-1)/2}}\Big(D_{r/2}(z;q)P_\kappa(z;q,q^{(r+1)/2})\Big) \notag \\ 
&\hspace{1cm}=q^{r(N-1)/2}D_{r/2}(z;q^{1/2})\sum_{i=1}^N\bigg( \prod_{\substack{j=1 \\ \not= i}}\frac{(q^{(r+1)/2}z_i-z_j)}{z_i-z_j}\bigg) \; T_{q,z_i}P_{\kappa}(z;q.q^{(r+1)/2}) \notag\\
&\hspace{1cm}=q^{r(N-1)/2}D_{r/2}(z;q^{1/2}) e(\kappa;q,q^{(r+1)/2})P_\kappa(z;q,q^{(r+1)/2}). \label{20.d}
\end{align}
But according to the definition (\ref{Mc})
$$
q^{r(N-1)/2}e(\kappa;q,q^{(r+1)/2})=e(\kappa;q,q^{(r-1)/2}),
$$
which establishes the sought eigenequation. \hfill $\square$
\end{proof}

We remark that the case $\kappa = 0^N$ of Proposition \ref{PD} was known previously
\cite[Thm.~3.2]{BL08}. Also, we draw attention to the work \cite{Lu10} in which a factorization identity
for $P_\kappa(z;q,q^k)$, with $z$ an infinite number of variables corresponding to the alphabet
${1 - q \over 1 - q^k} \mathbb X$, is established.

\subsection{Nonsymmetric Macdonald polynomials}
In the work \cite{Ka05c} vanishing properties of the nonsymmetric Macdonald polynomials with $t^{k+1}q^{r-1}=1$ were studied. It was found that for $k=1$ the relevant wheel condition is 
$$
z_j=z_it q^s,\qquad t=q^{-(r-1)/2}
$$
where $0\leq s \leq r-2$, and it is required $j<i$ if $s=0$, and $i\not=j$ otherwise. It follows that for $l$ odd
\begin{equation}\label{22.1}
E_{\kappa+l \delta}(z;q,q^{-l/2})=D_{(l-1)/2}(z;q)\prod_{1\leq i < j \leq N}(q^{(l-1)/2}z_j-z_i)f(z),
\end{equation}
where $f(z)$ is a homogeneous polynomial of degree $\kappa$. However, unlike the case of the Jack limit (\ref{J11}), computer algebra computations show that in general $f(z)$ is not itself a single nonsymmetric Macdonald polynomial.

\section{The clustering condition for $k>1$}
\subsection{Jack polynomials}
The case $k=1$ is special because the clustering condition determines the explicit form (\ref{J8.1}) of the wave function. This is not the case for $k>1$, and moreover the solution of (\ref{J8.1}) in terms of the Jack polynomials (\ref{J51}) appears not to be a result available from the existing Jack polynomial literature. Our first point of interest is to provide a self contained derivation from the perspective of Jack polynomial theory. This can be carried out as a corollary of the following general formula, the derivation of which is quite elementary.

\begin{proposition} \label{G}
Let the part $0$ in $\kappa=(\kappa_1,\ldots, \kappa_N)$ have frequency $f_0$ so that $l(\kappa)=N-f_0$, and set $\widetilde{\kappa}=(\kappa_1,\ldots,\kappa_{N-f_0})-(\kappa_{N-f_0})^{N-f_0}$. For the Jack polynomials $P_\kappa(z;\alpha)$ and $P_{\widetilde{\kappa}}(z;\alpha)$ let $\alpha$ be such that they are translationally invariant, and thus
\begin{align}
P_\kappa(z_1,\ldots, z_N;\alpha)&=P_{\widetilde{\kappa}}(z_1-1,\ldots, z_N-1;\alpha)\label{23.1}\\
P_{\widetilde{\kappa}}(z_1,\ldots, z_{N-f_0};\alpha)&= P_{\widetilde{\kappa}}(z_1-1,\ldots, z_{N-f_0}-1;\alpha)\label{23.2}.
\end{align}
We have 
\begin{align}
&P_\kappa(z_1,\ldots, z_N;\alpha)\bigg|_{z_{N-f_0+1}=\cdots=z_N=1} \notag \\
&\hspace{2cm}=\prod_{l=1}^{N-f_0}(1-z_l)^{\kappa_N-f_0}P_{\widetilde{\kappa}}(z_1,\ldots, z_{N-f_0};\alpha). \label{23.3}
\end{align}
\end{proposition}
\begin{proof}
We have 
\begin{align*}
&P_\kappa(z_1,\ldots,z_{N-f_0},\underbrace{1,\ldots,1}_{f_0\;\;times};\alpha) \\
&\hspace{2cm}=P_\kappa(z_1-1,\ldots,z_{N-f_0}-1,\underbrace{0,\ldots,0}_{f_0\;\;times};\alpha) \\
&\hspace{2cm}=P_{\widetilde{\kappa}+(\kappa_{N+f_0})^{N-f_0}}(z_1-1,\ldots,z_{N-f_0}-1;\alpha)\\
&\hspace{2cm}=\prod_{l=1}^{N-f_0}(z_l-1)^{\kappa_{N-f_0}}P_{\widetilde{\kappa}}(z_1-1,\ldots,z_{N-f_0}-1;\alpha)\\
&\hspace{2cm}=\prod_{l=1}^{N-f_0}(z_l-1)^{\kappa_{N-f_0}}P_{\widetilde{\kappa}}(z_1,\ldots,z_{N-f_0};\alpha),
\end{align*}
where the first equality follows from (\ref{23.1}), the second from the stability property of the Jack polynomials, the third from the simple property of Jack polynomials that $P_{\kappa+p^N}(z;\alpha)=z^{p^N}P_\kappa(z;\alpha)$ and the fourth from the assumption (\ref{23.2}).\hfill $\square$
\end{proof}

According to Proposition \ref{G}, in order to verify that (\ref{J51}) satisfies (\ref{J4}), it suffices to establish (\ref{23.1}) and (\ref{23.2}). But these in turn are immediate corollaries of the highest weight condition (\ref{Lp}) (thus (\ref{Lp}) implies $\psi$ must be a symmetric function in $z_i-\overline{z}$, $\overline{z}=\sum_{j=1}^Nz_j/N$; see e.g. \cite{Li10}), which we know from \cite{BH08a} is satisfied by (\ref{J4}).

In \cite{BH08} the class of partitions $\kappa$ for which $P_\kappa(z;-(k+1)/(r-1))$ satisfies the highest weight condition was extended from (\ref{J54}) to include a positive integer $s\geq1$ specified by
\begin{equation}\label{25x}
\kappa(k,r,s)=[n_00^{(r-1)s}k 0 ^{r-1}k 0 ^{r-1}k 0 ^{r-1} k \cdots]
\end{equation}
where $N - l(\kappa) = (k+1)s - 1 =: n_0$ (the case $s=1$ corresponds to (\ref{J54})).
This generalization was itself generalized in \cite{JL10} to include a further positive integer
$1 \le m \le k$ specified by
\begin{equation}\label{25xx}
\kappa(k,r,s,m)=[n_00^{(r-1)s}k 0 ^{r-1}k 0 ^{r-1}k   \cdots 0 ^{r-1} m].
\end{equation}
It follows from Proposition \ref{G} that with $\alpha=-(k+1)/(r-1)$, $(k+1)$ and $(r-1)$ relatively prime
\begin{align}
&P_{\kappa(k,r,s,m)}(z_1,\ldots, z_N;\alpha)\bigg|_{z_{N-n_0+1}=\ldots=z_N=z}\notag \\
&\hspace{2cm}=\prod_{j=1}^{N-n_0}(z_j-z)^{(r-1)s+1}P_{\kappa(k,r,1,m)}(z_1,\ldots,z_{N-n_0})\label{25.1}
\end{align}
(homogeneity of the Jack polynomials has been used to go from setting $z_{N-n_0+1}=\ldots=z_N=1$ to setting $z_{N-n_0+1}=\ldots=z_N=z$).

The partition (\ref{25x}), written therein in terms of occupation numbers, is an example of a staircase partition. The special case of staircase partitions, consisting of a single part $r$ repeated say $g$ times is referred to as a rectangular partition. It was shown in \cite{JL10} that with 
\begin{equation}\label{25.2}
\alpha=-\frac{N+1-g}{r-1},\qquad N\geq 2g
\end{equation}
and $N+1-g$, $r-1$ relatively prime, the Jack polynomial $P_{r^{g}}(z;\alpha)$ satisfies the highest weight condition. It then follows from Proposition \ref{G} that
\begin{equation}\label{26}
P_{r^g}(z_1,\ldots,z_N;-(N+1-g)/(r-1))\bigg|_{z_{g+1}=\ldots=z_N=1}=\prod_{l=1}^{g}(z_l-1)^r
\end{equation} 

We now turn our attention to the case of nonsymmetric Jack polynomials in the context of the clustering condition (\ref{25.1}). Unless $s=1$ and $m=k$, the nonsymmetric counterparts of the Jack polynomials in (\ref{25.1}) do not satisfy the highest weight condition and so Proposition \ref{G} does not apply. In fact the computer algebra computations show that there is no longer a factorization of the same type as (\ref{25.1}). We find instead a structure
\begin{equation}\label{26.1}
E_{\kappa(k,r,s,m)}(z_1,\ldots,z_N;\alpha)\bigg|_{z_{N-n_0+1}=\ldots=z_N=z}=\prod_{j=1}^{N-n_0}(z_j-z)^{(r-1)s}f(z,z_1,\ldots, z_{N-n_0}).
\end{equation} 
This differs from (\ref{25.1}) in that the exponent in the first lot of factorized products is reduced by $1$, and the variable $z$ becomes a part of the remaining factor $f$.

\subsection{Generalized Hermite and Laguerre polynomials}
In Jack polynomial theory the (symmetric) binomial coefficients can be specified by the generalized binomial expansion \cite[eq.~(12.179)]{Fo10}
\begin{equation}\label{B1}
{P_\kappa(1+x;\alpha) \over P_\kappa((1)^N;\alpha)} =
\sum_{\mu \subseteq \kappa} \Big ( { \kappa \atop \mu} \Big )
{P_\mu(x;\alpha) \over  P_\mu((1)^N;\alpha)}.
\end{equation} 
It is known \cite{BF97b} that the symmetric Laguerre polynomials can be expanded in terms of the
symmetric binomial coefficients according to
\begin{eqnarray}\label{B2}
&&P_\kappa^{(L)}(x;\alpha) = (-1)^{|\kappa|}   P_\kappa((1)^N;\alpha) [a+h]_\kappa^{(\alpha)}
\sum_{\mu \subseteq \kappa} \Big ( { \kappa \atop \mu} \Big )
{ (-1)^{|\mu|}  P_\mu(x;\alpha) \over   [a+h]_\mu^{(\alpha)} P_\mu((1)^N;\alpha)},
\end{eqnarray} 
where $h := 1 + (N-1)/\alpha$ and
$$
[u]_\kappa^{(\alpha)} := \prod_{j=1}^N {\Gamma(u - (j-1)/\alpha + \kappa_j) \over
\Gamma(u - (j-1)/\alpha)}.
$$

Suppose now that $\alpha$ and $\kappa$ are such that $P_\kappa(x;\alpha)$ satisfies the highest weight condition. We then have $P_\kappa(1+x;\alpha) = P_\kappa(x;\alpha)$, and this in
turn implies that $ P_\kappa((1)^N;\alpha) = 0$. Moreover, if follows from (\ref{B1}) that we also have
$$
 P_\kappa((1)^N;\alpha)   \Big ( { \kappa \atop \mu} \Big )
{P_\mu(x;\alpha) \over  P_\mu((1)^N;\alpha)} = 0, \qquad \mu \ne \kappa.
$$
Using this in (\ref{B2}), together with the fact that $[a+h]_\kappa^{(\alpha)}/[a+h]_\mu^{(\alpha)}$
must be finite for $\mu \subseteq \kappa$ we see that all terms in (\ref{B2}) must vanish except for the
term $\mu = \kappa$. Thus in this setting we have in general
\begin{equation}\label{B3}
P_\kappa^{(L)}(x;\alpha) = P_\kappa(x;\alpha).
\end{equation} 
Also, by inspection of the orthogonalities of the generalized Hermite and Laguerre polynomials
(see e.g.~\cite[Ch.~13]{Fo10}, it is easy to see that
\begin{equation}\label{B4}
\lim_{a \to \infty} \Big ( {1 \over \sqrt{2a}} \Big )^{|\kappa|}
P_\kappa^{(L)}(a + \sqrt{2a} x;\alpha) = P_\kappa^{(H)}(x;\alpha).
\end{equation}
Suppose  $\alpha$ and $\kappa$ are such that $P_\kappa(x;\alpha)$ satisfies the highest weight condition. Then (\ref{B3}) holds and it substituted in (\ref{B4}) implies
\begin{equation}\label{B5}
P_\kappa^{(H)}(x;\alpha) = P_\kappa(x;\alpha).
\end{equation} 
The identities (\ref{B3}) and (\ref{B5}) tell us that the symmetric generalized   Hermite and Laguerre polynomials satisfy the analogues of (\ref{25.1}) and (\ref{26}).

With regards to the nonsymmetric Hermite and Laguerre polynomials, computer algebra computations indicate a factorization having the structure (\ref{26.1}), but again we have not been able to identify the analogue of the function $f$.

\subsection{Macdonald polynomials}
Set $\alpha = - (k+1)/(r-1)$, $(k+1)$ and $(r-1)$ relatively prime and require that
$z_i = q^{(i-1)/\alpha} z$, $(i=1,\dots,(k+1)s - 1)$. Computer algebra computations indicate that the
Macdonald polynomial analogue of the identity (\ref{25.1}) is
\begin{align}
&P_{\kappa(k,r,s,m)}(z_1,\ldots, z_N;q,q^{1/\alpha})\bigg|_{z_{i}= q^{(i-1)/\alpha} z \: \: (i=1,\dots,
(k+1)s - 1)}\notag \\
&\hspace{2cm}=
\prod_{j=-(r-1)(s-1)}^{r-1}
\prod_{i=s(k+1)}^{N}(z_i-zq^{k/\alpha + j} ) 
P_{\kappa(k,r,1,m)}(z_{s(k+1)},\ldots,z_{N};q,q^{1/\alpha})\label{23.8}.
\end{align}
In the case $k=s=m=1$ we see that this is consistent with the case $\kappa = 0^N$ of (\ref{20.c}).
For $s=1$, $m=k$, the wheel condition \cite{FJMM03} tells us that with
$$
z_1 = x, \quad z_2 = tx_1, \quad \dots, \quad z_k = t^{k-1} x,
$$
the Macdonald polynomial $P_{\kappa(k,r,s,m)}(z;q,q^{1/\alpha})$ must vanish for
$z_j = t^k q^s x$ ($s=0,\dots,r-1$). We see that conjectured identity (\ref{23.8}) is consistent
with this requirement. Furthermore, in the limit $q \to 1$  (\ref{23.8}) reduces to the Jack polynomial
identity (\ref{25.1}).

Our derivation of (\ref{25.1}) relied on the corresponding Jack polynomials satisfying the highest weight condition. The significance of the latter being that it implied the translation invariance conditions (\ref{23.1}), (\ref{23.2}). Let
$$
{\partial \over \partial_q z_i} f(z) :=
{f(z) - T_{q,z_i} f(z) \over (1 - q) z_i}.
$$
The analogue of the highest weight condition for the Macdonald polynomials, $\psi$ say, in (\ref{23.8}) is
\cite{JL10}
\begin{equation}\label{Lqt}
L^{+(q,t)} \psi = 0, \qquad
L^{+(q,t)} := \sum_{i=1}^N \Big ( \prod_{j=1 \atop \ne i}^N 
{t z_i - z_j \over z_i - z_j} \Big ) {\partial \over \partial_q z_i}
\end{equation} 
(compare the definition of $L^{+(q,t)}$ with the definition of $M_1$ in (\ref{Mb})). But as yet we have not seen how to make use of this property in providing a proof of (\ref{23.8}).

It is shown in \cite{JL10} that the Macdonald polynomials $P_{r^g}(z;q,q^{1/\alpha})$, with
$\alpha$ as in (\ref{25.2}), also satisfy the $(q,t)$-highest weight condition (\ref{Lqt}).
Computer algebra computations indicate that
$$
P_{r^g}(z,zq^{1/\alpha},\dots, zq^{(N-g-1)/\alpha},z_{N-g+1},\dots,z_N) =
\prod_{l=N-g+1}^N \prod_{j=0}^{r-1} (z_l - q^{1/\alpha + j}z),
$$
although as with (\ref{23.8}) we are yet to find a proof.

\section{Concluding remarks}
The focus of our study has been the setting within Jack polynomial theory --- interpreted broadly to include generalized Hermite and Laguerre polynomials, and Macdonald polynomials --- of the $(k,r)$ clustering condition (\ref{J4}) and its generalization (\ref{25.1}). In the case $k=1$ we were able to identify more general identities in Jack and Macdonald polynomial theory containing (\ref{J4}) as a special case. For general $(k,r)$, by making essential use of the highest weight condition
(\ref{Lp}), we found a simple derivation of the Jack polynomial solution to (\ref{J4}), and more
generally of the clustering property (\ref{25.1}). However our derivation does not apply to the Macdonald analogue (\ref{23.8}) of this clustering, formulated on the basis of computer algebra computations,
which remains a conjecture. We should point out that there is interest in clustering/vanishing properties of Macdonald polynomials with $t^{k+1} q^{r-1}=1$
for there application to certain statistical mechanical models based on the Temperley-Lieb
algebra \cite{KP07, RSZ07, dGPS09}

To finish, we make note of a $(q,t)$-generalization of the Reed-Rezayi state (\ref{RR}),
$$
\psi_{RR}^{(k;q,t)}(z_1,\dots,z_{kN}) = U^+
\prod_{s=1}^k \prod_{1 \le i_s < j_s \le N} (z_{i_s} - t z_{j_s}) (t z_{i_s} - z_{j_s})
\Big |_{t = q^{-1/(k+1)}},
$$
where $U^+$ is the $t$-symmetrization operator in (\ref{UU}). Thus computer algebra computations indicate that
$$
\psi_{RR}^{(k;q,t)}(z_1,\dots,z_{kN}) \propto
P_{\kappa(k,2,1)}(z_1,\dots,z_{kN};q,q^{-1/(k+1)})
$$
as is consistent with (\ref{23.8}). We remark that $(q,t)$-generalizations of quantum Hall states seem
first to have been considered in the work of Kasatani and Pasquier \cite{KP07}.

\subsection*{Acknowledgements}
This work was supported by the Australian Research Council.

%\bibliographystyle{amsplain}
%\bibliography{book1}

\providecommand{\bysame}{\leavevmode\hbox to3em{\hrulefill}\thinspace}
\providecommand{\MR}{\relax\ifhmode\unskip\space\fi MR }
\providecommand{\MRhref}[2]{%
  \href{http://www.ams.org/mathscinet-getitem?mr=#1}{#2}
}
\providecommand{\href}[2]{#2}

\end{document}